\documentclass[aps,pra,reprint,10pt,a4paper,nofootinbib,showpacs]{revtex4-2}
\usepackage[utf8]{inputenc}
\usepackage{amsfonts}
\usepackage{amssymb}
\usepackage{amsmath}
\usepackage{amsthm}
\usepackage{graphics}
\usepackage{graphicx}
\usepackage{braket}
\usepackage[colorlinks=true,linkcolor=blue,urlcolor=blue,citecolor=blue]{hyperref}

\newtheorem{proposition}{Proposition}

\begin{document}
\title{Unextendible entangled bases and more nonlocality with less entanglement}

\author{Saronath Halder}
\affiliation{Harish-Chandra Research Institute, HBNI, Chhatnag Road, Jhunsi, Allahabad 211 019, India}

\author{Ujjwal Sen}
\affiliation{Harish-Chandra Research Institute, HBNI, Chhatnag Road, Jhunsi, Allahabad 211 019, India}

\begin{abstract}
We consider a general version of the phenomenon of more nonlocality with less entanglement, within the framework of the unambiguous (i.e., conclusive) quantum state discrimination problem under local quantum operations and classical communication. We show that although the phenomenon was obtained before for two qutrits, it can also be observed for two qubits, while still being at the single-copy level. We establish that the phenomenon is intrinsically connected to the concept of unextendible entangled bases, in the two-qubit case. In the process, we demonstrate a hierarchy of nonlocality among sets of two-qubit orthogonal pure states, where the ``nonlocality'' is in the sense of a difference between global and local abilities of quantum state discrimination. We present a complete characterization of two-qubit pure orthogonal state sets of cardinality three with respect to their nonlocality in terms of unambiguous local distinguishability, the status for other cardinalities being already known. The results are potentially useful for secure quantum communication technologies with an optimal amount of resources.
\end{abstract}
\maketitle

\section{Introduction}\label{sec1}
A focus of current quantum technologies is to control quantum systems for implementation of  information processing protocols \cite{Nielsen00}, the  reason being that a quantum device  may provide advantage over its classical counterpart \cite{Shor95, Buhrman01, Guerin16, Harrow18, Zhong20}. An important step in information processing protocols is to decode the information encoded within the state of a quantum system. For this purpose, it may be necessary to distinguish among the possible states of the given system. Secure distribution of information among several spatially separated parties potentially requires employing ensembles of quantum states, information encoded in which is not easy for dishonest parties or outsiders to decode.

When a composite quantum system is distributed among several spatially separated parties, it may not be possible to distinguish among the possible states of the system perfectly if only local quantum operations and classical communication (LOCC) are allowed, even though the states are pairwise orthogonal. Formally, the problem of state discrimination under LOCC (also known as the local state discrimination problem) is to optimally identify, by LOCC, the state that is secretly chosen from a given set.

The fact that separable states can be created by LOCC \cite{Werner89} while entangled ones \cite{Horodecki09-1} cannot, may coax one to presume that the local indistinguishability of an ensemble of orthogonal states of a multiparty system is due to entanglement present in the ensemble states. While this is indeed true to some extent \cite{Chen03, Chen03-1, Badziag03, Horodecki04, Ghosh04, Ghosh05, Hayashi06, Sen(De)06, Matthews09}, an actual deciphering of the phenomenon is still elusive and has roots independent of entanglement \cite{Horodecki07}. A seminal result in this direction is the discovery of a set of orthogonal \emph{product} states which cannot be perfectly distinguished by LOCC \cite{Bennett99-1}. Another landmark work demonstrates that two orthogonal pure states can always be distinguished by LOCC, irrespective of their entanglement content~\cite{Walgate00}. Thereafter, several related results have been reported in the literature~\cite{Virmani01, Ghosh01, Groisman01, Walgate02, Ghosh02, Horodecki03, Badziag03, Ghosh04, Rinaldis04, Fan04, Nathanson05, Watrous05, Niset06, Hayashi06, Duan07, Feng09, Bandyopadhyay11, Yu12, Yang13, Childs13, Zhang14, Zhang15, Xu16, Croke17, Halder18, Halder19}. The importance of the study of local indistinguishability of quantum states stems from the key requirement in many quantum information protocols, such as data hiding \cite{Terhal01, Eggeling02} and secret sharing \cite{Markham08}, to distinguish a quantum ensemble locally. Furthermore, the setting of the local state discrimination problem can be useful to demonstrate ``nonlocal'' properties of composite quantum systems, where the term, ``nonlocality'', is used to signify a difference between global and local distinguishabilities of an ensemble of quantum states.

A prominent set of examples of LOCC indistinguishable sets of orthogonal quantum states is provided by the unextendible product bases~\cite{Bennett99, Divincenzo03, Rinaldis04}. Subsequently, unextendibility for entangled states has also been studied~\cite{Bravyi11, Chen13, Li14, Wang14, Nan15, Nizamidin15, Guo16, Zhang16-3, Wang17-1, Zhang18-1, Zhang18-2, Liu18, Song18, Zhao20, Guo14, Han18, Shi19, Yong19, Wang19, Chakrabarty12,Chen13-1, Guo15-1, Zhang17-4, Halder21}, although the LOCC (in)distinguishability properties of such bases are less explored~\cite{Halder21}. It is potentially useful to provide a brief survey of the results related to unextendibility of sets of entangled states. Unextendible entangled bases using orthogonal maximally entangled states were introduced in Ref.~\cite{Bravyi11} for $\mathbb{C}^d\otimes\mathbb{C}^d$, where $d=3,4$. Subsequently, unextendible maximally entangled bases were presented in \cite{Chen13} for $\mathbb{C}^{d_1}\otimes\mathbb{C}^{d_2}$,  where $d_1d_2>4$, $d_2/2<d_1<d_2$. Other articles discussing bipartite unextendible maximally entangled bases include~\cite{Li14, Wang14, Nan15, Nizamidin15, Guo16, Zhang16-3, Wang17-1, Zhang18-1, Zhang18-2, Liu18, Song18, Zhao20}. Unextendible entangled bases for fixed Schmidt ranks were presented in Refs.~\cite{Guo14, Han18, Shi19, Yong19, Wang19}. A type of unextendibility for nonmaximally entangled states was presented in Refs.~\cite{Chakrabarty12,Chen13-1}. Unextendibility for multipartite entangled states was discussed in Refs.~\cite{Guo15-1, Zhang17-4}. Further, in Ref.~\cite{Halder21}, local distinguishability and indistinguishability properties of both bipartite and multipartite unextendible entangled bases were discussed. 

If a given set of orthogonal quantum states cannot be perfectly distinguished by LOCC, then it is usual to claim that the ensemble possesses ``nonlocality''. If perfect discrimination is not possible, one then looks for unambiguous discrimination under LOCC~\cite{Chefles98, Chefles04, Ji05, Duan07, Walgate08, Bandyopadhyay09, Cohen14}. This setting is interesting because with some nonzero probability, it is possible to distinguish the states without committing any error. It has also been referred to in the literature as the conclusive quantum state discrimination problem. As mentioned earlier, local indistinguishability, and in particular the unambiguous case, is a key ingredient for many information processing protocols, and therefore important for building secure quantum technologies. 

A phenomenon that further dissociates local indistinguishability of multiparty orthogonal quantum states with entanglement-like concepts is what has been termed ``more nonlocality with less entanglement''~\cite{Horodecki03}. It demonstrates the possibility to get a locally indistinguishable set of orthogonal states by replacing a more entangled state with a less entangled one in a locally distinguishable set. The phenomenon was obtained in a system of two qutrits\footnote{A qutrit is a  three-dimensional quantum system.}, within the scenario of deterministic (i.e., with unit probability) distinguishability under LOCC-based measurements. The example of course helps in underlining that there is more to local indistinguishability of orthogonal states than the entanglement content of the ensemble states. However, it is an isolated example, and to understand the phenomenon and obtain  qualitative and quantitative comprehension of it, we need to analyze further instances and connect the occurrences of the phenomenon with other phenomena in quantum information and possibly beyond. This is even more important given that there are indications that local indistinguishability is also related with entanglement to some extent \cite{Chen03, Chen03-1, Badziag03, Horodecki04, Ghosh04, Ghosh05, Hayashi06, Sen(De)06, Matthews09}. The phenomenon more nonlocality with less entanglement is particularly important because of its potential ability to provide local indistinguishability with an ensemble with reduced average entanglement than other candidates. 
 
In this work, we establish connections between two \emph{a priori} far-flung concepts of quantum information, viz. unextendible entangled bases and the phenomenon of more nonlocality with less entanglement. Unlike in the original work, we are able to demonstrate the phenomenon already for two qubits, the minimal-dimensional multiparty system. We present here a general version of the above phenomenon using the setting of unambiguous local state discrimination (unlike the deterministic scenario considered in Ref.~\cite{Horodecki03}). Therefore, the present version of the phenomenon provides  potentially more ensembles that can be employed for secure information processing, while the cost of preparing such ensembles might be less than what was required before. We further show that unambiguous local state discrimination is useful to examine the degree of nonlocality of ensembles that cannot be perfectly distinguished by LOCC. For the case of two qubits, we obtain a complete characterization of the possible combinations which may occur within the setting of unambiguous state discrimination under LOCC. 

\section{Preliminaries}\label{sec2}
{\it Unambiguous discrimination.---}~Given a set of states $\{\ket{\psi_i}\}_i$, if we pick a particular state $\ket{\psi_i}$, then the state can be unambiguously identified if and only if it is possible to recognize the state with some nonzero probability without committing any error \cite{Chefles00, Barnett09, Bergou10, Bae15}. Moreover, if every state of the given set can be unambiguously identified then we say that the set is unambiguously distinguishable. In this work, we consider only the two-qubit quantum system, which is distributed among two spatially separated parties and they are allowed to perform LOCC only. So, if a given set of this physical system is unambiguously distinguishable under LOCC, then we say that the set is unambiguously locally distinguishable or the set is conclusively locally distinguishable. If a set cannot be distinguished unambiguously by LOCC, then it means that the set contains at least one state which cannot be unambiguously identified by LOCC. For unambiguous discrimination, it is necessary to consider linearly independent quantum states~\cite{Chefles98}.

{\it More nonlocality with less entanglement.---}~A set of pure orthogonal entangled states was reported in Ref.~\cite{Horodecki03} that is perfectly locally distinguishable, but in which if we replace one entangled state  by a product state, then the new set cannot be perfectly distinguished by LOCC. This appearance of local indistinguishability with lowering of average entanglement of the ensemble states was termed as more nonlocality with less entanglement. In the present work, we consider several sets that are nonlocal in the sense that they cannot be perfectly distinguished by LOCC. So, the sets are equally nonlocal from the perspective of perfect discrimination by LOCC. Still, it is possible to put a hierarchy among those sets using the setting of unambiguous discrimination under LOCC. Moreover, the degree of nonlocality (in the sense of the strength of local indistinguishability) may increase if we decrease the average entanglement content of the set or more strikingly if we replace an entangled state by a product state. It is in this sense that we claim to have considered a general version of the phenomenon of more nonlocality with less entanglement. We mention here that for any state discrimination problem considered in this paper, the given states are equally probable and only a single copy of each state of the given set is available. In this context, note that the phenomenon in the multi-copy limit was considered in Ref.~\cite{Banik21}. See also Ref.~\cite{Sen(De)12} in this regard.

{\it Chefles's criterion.---}~For unambiguous identification under LOCC, a necessary and sufficient condition has been derived in Ref.~\cite{Chefles04}. A simplified version of the condition can be found in Ref.~\cite{Bandyopadhyay09}. According to the condition, when a set of states $\{\ket{\psi_i}\}_i$ is given, to identify a particular state $\ket{\psi_i}$ among them unambiguously by LOCC, it is necessary and also sufficient that there exists at least one product state $\ket{\alpha}$ such that $\langle\psi_i|\alpha\rangle>0$ and for all $j \neq i$, $\langle\psi_j|\alpha\rangle=0$.
 
{\it Unextendible entangled basis.---}~While the concept is generic to all multiparty quantum systems, we will need it only for bipartite systems. Consider a set of orthogonal pure entangled states of a bipartite tensor-product Hilbert space. The states can be maximally or nonmaximally entangled states. If the states span a proper subspace of the given Hilbert space such that the complementary subspace contains no entangled state, then the considered set is said to form an unextendible entangled basis (UEB).

\section{Results}\label{sec3}
Two orthogonal pure states can always be perfectly distinguished by LOCC \cite{Walgate00}. A complete orthonormal basis cannot be unambiguously distinguished locally if and only if the set contains at least one entangled state~\cite{Horodecki03,Chefles04}.\footnote{As mentioned in Ref.~\cite{Horodecki03}, a complete orthonormal basis can never contain a single entangled state, as otherwise the remaining states will obviously form an unextendible product basis, which is not allowed, because that will lead to the pure entangled state to be bound entangled \cite{Horodecki97, Horodecki98}.} For two-qubit systems, this was already shown in Ref.~\cite{Walgate02}. If a complete basis contains $n$ entangled states, then these entangled states cannot be unambiguously identified by LOCC, and based on this fact, it is possible to place a simple hierarchy among the complete bases with different values of $n$.

On the other hand, if a set contains three orthogonal two-qubit pure states, then the situation is far richer, and there are several interesting cases with respect to their local distinguishability. In this regard, we consider only those sets which contain at least two entangled states, as otherwise the sets are perfectly distinguishable by LOCC~\cite{Walgate02}.\footnote{For unambiguous discrimination of three nonorthogonal states, see Ref.~\cite{Bandyopadhyay09}.} We begin by presenting the following proposition for any set of three two-qubit orthogonal {\it maximally} entangled states.

\begin{proposition}\label{prop1}
Any set of three two-qubit pairwise orthogonal maximally entangled states is unambiguously distinguishable by LOCC.
\end{proposition} 

\begin{proof}
Consider an ensemble of three {\it maximally} entangled orthogonal two-qubit states $\ket{\phi_1}$, $\ket{\phi_2}$, and $\ket{\phi_3}$. This implies that the state $\ket{\phi_4}$ which is orthogonal to the states $\ket{\phi_1}$, $\ket{\phi_2}$, and $\ket{\phi_3}$, must also be a maximally entangled state~\cite{Bravyi11}. Now, we choose any state $\ket{\phi_i}$, $i\in\{1,2,3\}$. The states $\ket{\phi_i}$ and $\ket{\phi_4}$ span a subspace which contains at least one product state~\cite{Sanpera98}. This product state has nonzero overlaps with both the states $\ket{\phi_i}$ and $\ket{\phi_4}$, as the two latter states are themselves non-product. Furthermore, the product state must be orthogonal to the states $\ket{\phi_j}$, $j \in \{1, 2, 3\} \setminus \{i\}$. So, we obtain that every $\ket{\phi_i}$ can be unambiguously identified by LOCC with some nonzero probability (recall that we begin with any $\ket{\phi_i}$, $i\in\{1,2,3\}$).
\end{proof}

We next construct a class of two-qubit UEBs which are sets of three orthogonal nonmaximally entangled states. The states in the sets are given by

\begin{equation}\label{eq1}
\begin{array}{l}
\ket{\psi_1} = \sqrt{\lambda_1}\ket{01} + \sqrt{\lambda_2}\ket{10},\\ [0.5 ex]
\ket{\psi_2} = \sqrt{\lambda_3}\ket{00} + \sqrt{\lambda_4}\ket{\psi_1^\perp},\\ [0.5 ex]
\ket{\psi_3} = \sqrt{\lambda_4}\ket{00} - \sqrt{\lambda_3}\ket{\psi_1^\perp},
\end{array}
\end{equation} 
where $\ket{\psi_1^\perp}$ = $\sqrt{\lambda_2}\ket{01} - \sqrt{\lambda_1}\ket{10}$. The $\lambda_i$ are positive numbers and $\lambda_1+\lambda_2$ = 1 = $\lambda_3 + \lambda_4$. The  states, \(|\psi_i\rangle\), are all nonmaximally entangled states, but are not equally entangled. The only state orthogonal to the \(|\psi_i\rangle\) is $\ket{11}$, a product state. So, the above set of states form an UEB, and for varying parameters, form a class of UEBs. 

\begin{proposition}\label{prop2}
Any UEB formed of  the states in (\ref{eq1}) is not unambiguously distinguishable under LOCC.
\end{proposition}

\begin{proof}
If an UEB consisting the states in (\ref{eq1}) is unambiguously distinguishable by LOCC, then every state in it will be unambiguously identifiable by LOCC. In particular, $\ket{\psi_1}$ will be unambiguously identified by LOCC, from among the states in the UEB. Then, there will exist a product state which has nonzero overlap with $\ket{\psi_1}$, whereas the same product state must have zero overlap with the other two states $\ket{\psi_2}$ and $\ket{\psi_3}$. Clearly, the product state must belong to the subspace spanned by $\ket{\psi_1}$ and $\ket{11}$. But it is possible to show explicitly that this subspace contains only one product state, viz. $\ket{11}$, which is orthogonal to $\ket{\psi_1}$. This implies that \(|\psi_i\rangle\) cannot be unambiguously locally distinguished when the state sent is promised to be from the set \(\{|\psi_i\rangle\}_{i=1}^3\). This is a contradiction, implying that the original assumption that \(\{|\psi_i\rangle\}_{i=1}^3\) can be unambiguously distinguished by LOCC was not correct. Hence, the proof. 
\end{proof}

Consider now  the following points that we gather from the preceding two propositions.
\begin{itemize}
\item {\it More nonlocality in sets of Proposition~\ref{prop2} than those in Proposition~\ref{prop1}:}  The sets of states considered in the Propositions \ref{prop1} and \ref{prop2} are equally nonlocal when considered with respect to  perfect discrimination by LOCC, as in both cases, they cannot be perfectly distinguished by LOCC~\cite{Walgate02}. But the setting of unambiguous state discrimination under LOCC provides us the privilege to put a hierarchy among the sets. Precisely, we can claim that the sets of Proposition \ref{prop2} are \emph{more nonlocal} compared to those of Proposition \ref{prop1}, because the former are not unambiguously distinguishable under LOCC while the latter are.

\item {\it Less entanglement in sets of Proposition~\ref{prop2} than those in Proposition~\ref{prop1}:} The average entanglement contents of the sets\footnote{The average entanglement content of a set is the mean of the amounts of entanglement contained in the states of the set.} in Proposition \ref{prop2} are strictly less than the average entanglement content of the sets of Proposition \ref{prop1}, as the latter contains  maximally entangled states, while the former does not.
\end{itemize}

These two points therefore provide us with an instance of the phenomenon of more nonlocality with less entanglement. Unlike in Ref.~\cite{Horodecki03}, the instance is in the lowest multiparty physical system. Also unlike in that reference, the phenomenon is obtained using unambiguous local discrimination instead of the perfect variety. More specifically, in Ref.~\cite{Horodecki03}, perfect discrimination of quantum states is considered while in our case, we consider the probabilistic discrimination of quantum states.

{\it Application.---}~As an application of the result obtained, we identify its potential use in the following information processing task. Suppose that there is a trit of classical information that a sister (Enola) wants to send to her two brothers (Mycroft and Sherlock)~\cite{Springer06}, in such a way that the brothers have to come together to find the information: deterministic LOCC  between Mycroft and Sherlock is not enough here, where the LOCC is to act on the bipartite quantum states on which the classical trit is encoded by Enola and sent to the brothers.  Moreover, Enola wishes to encode one part of the information, say the 0 (of the trit of 0, 1, 2), which she judges as more important, in such a way that Mycroft and Sherlock won't be able to find it even by unambiguous LOCC with any nonzero probability. Proposition \ref{prop2} tells us that Enola is able to attain the feat by encoding the trit onto the elements of the UEBs of (\ref{eq1}), with the important part being encoded in the state $\ket{\psi_1}$ of (\ref{eq1}).

One may argue that if  linearly dependent states are used for the above protocol, then also those states cannot be unambiguously distinguished by LOCC. But in that case it is never possible to decode the encoded information perfectly.

{\it UEB necessary.---}~We present here a necessary condition related to the phenomenon of more nonlocality with less entanglement. For the demonstrated instance of the phenomenon, it is necessary that the three entangled states of (\ref{eq1}), form a UEB, because if they do not, then the fourth state $\ket{\psi_4}$, which is orthogonal to the states of (\ref{eq1}), can be an entangled state. Thus, while picking any state $\ket{\psi_i}$, $i\in\{1,2,3\}$, it is possible to find a product state in the span of $\ket{\psi_i}$ and $\ket{\psi_4}$, which is obviously nonorthogonal to $\ket{\psi_i}$ but orthogonal to $\ket{\psi_j}$, where $j \in \{1, 2, 3\} \setminus \{i\}$. Clearly, in this situation the considered set is not able to exhibit more nonlocality than the sets of Proposition~\ref{prop1}, with respect to unambiguous local discrimination.

{\it Best-case of more nonlocality with less entanglement.---}~We further consider sets which contains two entangled states and one product state, given by

\begin{equation}\label{eq2}
\begin{array}{l}
\ket{\Psi_1} = \ket{00},\\ [0.5 ex]
\ket{\Psi_2} = \sqrt{\lambda_1}\ket{01} + \sqrt{\lambda_2}\ket{10},\\ [0.5 ex]
\ket{\Psi_3} = \sqrt{\lambda_2}\ket{01} - \sqrt{\lambda_1}\ket{10},
\end{array}
\end{equation} 
where all $\lambda_i$ are positive numbers and $\lambda_1+\lambda_2$ = 1. A version of the above sets with nonorthogonal states can be found in Ref.~\cite{Bandyopadhyay09}. The  set, $\{\ket{\Psi_i}\}_i$, is not unambiguously distinguishable by LOCC. Moreover, both the states $\ket{\Psi_2}$ and $\ket{\Psi_3}$ are not unambiguously identifiable by LOCC. This follows by using the same line of proof as given for Proposition \ref{prop2}. Take a state $\ket{\Psi_i}$, $i\in\{2,3\}$. There is only one product state in the span of $\ket{\Psi_i}$ and $\ket{11}$, and the state is $\ket{11}$, which is orthogonal to $\ket{\Psi_i}$. Therefore, it is not possible to identify the state $\ket{\Psi_i}$ unambiguously under LOCC. 

Notice that the sets corresponding to Proposition \ref{prop1} contain no state which cannot be unambiguously identified by LOCC. The sets corresponding to Proposition \ref{prop2} contain only one state which cannot be unambiguously identified by LOCC. Both these types of sets (corresponding to Propositions \ref{prop1} and \ref{prop2}) contain only entangled states, but the set \(\{|\Psi_i\rangle\}_i\) (\(i=1,2,3\)) contain two entangled states and one product state. So, the number of entangled states has got reduced, and yet the set contains two states which cannot be unambiguously identified by LOCC. Furthermore, there is no set of three orthogonal two-qubit pure states which have all three states, not unambiguously identifiable by LOCC~\cite{Bandyopadhyay09}. In this way, we argue that the above sets $\{\ket{\Psi_i}\}_i$ exhibit the ``best-case'' of the phenomenon of more nonlocality with less entanglement, when the cardinality of the set of pure two-qubit orthogonal states is three. We present this statement in the following proposition.

\begin{proposition}\label{prop3}
The states of the set $\{\ket{\Psi_i}\}_i$ exhibit the best-case of more nonlocality with less entanglement sets of for two-qubit pure states of cardinality three.
\end{proposition}

Note that the set of states of (\ref{eq2}) is equally nonlocal with respect to the sets corresponding to the Propositions \ref{prop1} and \ref{prop2} when the perfect local discrimination of the states is considered.  But the setting of unambiguous discrimination under LOCC, again provides us the privilege of proving that the set of states of (\ref{eq2}) is more nonlocal in comparison to the sets corresponding to the Propositions \ref{prop1} and \ref{prop2}.

{\it UEB  necessary again.---}~Like a previous necessary condition, we now provide another necessary condition related now to the best-case of more nonlocality with less entanglement for sets of two-qubit pure states with cardinality three. It is necessary that the span of the states of (\ref{eq2}), is also spanned by a two-qubit UEB. Here it is important to mention that there is only one cardinality possible for two-qubit UEBs, and that is three~\cite{Halder21}. The rest of the proof follows the same line of argument as given in case of the previous necessary condition. If the span of the states of (\ref{eq2}) is not spanned by a two-qubit UEB, then the fourth state $\ket{\Psi_4}$, which is orthogonal to the states of (\ref{eq2}), can be an entangled state. Thus, while picking any state $\ket{\Psi_i}$, $i\in\{1,2,3\}$, it is possible to find a product state in the span of $\ket{\Psi_i}$ and $\ket{\Psi_4}$, which is obviously nonorthogonal to $\ket{\Psi_i}$ but orthogonal to $\ket{\Psi_j}$, where $j \in \{1, 2, 3\} \setminus \{i\}$. Clearly, in this situation, the considered set is not able to exhibit the ``best-case'' of more nonlocality with less entanglement. One such UEB that spans the subspace spanned by the states of (\ref{eq2}) is given by the set of states in (\ref{eq1}). Interestingly, the states of (\ref{eq1}) and those of (\ref{eq2}) span the same subspace, and yet the latter states exhibit more nonlocality, compared to the former ones.

For both the cases discussed in this paper, one may attempt to quantify nonlocality. This can be done, e.g.,  by considering the optimal average probability of success to identify the states unambiguously. However, it is a difficult quantity to actually evaluate, when the operational paradigm is LOCC, because of optimizations required over LOCC-based measurement strategies. Quantification of ``nonlocality'' in the sense of difference between global and local distinguishing abilities for arbitrary ensembles has been attempted in the literature (see e.g.~\cite{Horodecki07}), but is as yet not fully understood and the existing quantifications are difficult to evaluate.

\section{Conclusion}\label{sec4}
We have established  connections between the concept of unextendible entangled bases and the phenomenon of more nonlocality with less entanglement, for the elementary bipartite quantum system, i.e., a quantum system  associated with the two-qubit Hilbert space when only single copy of each input state is available. In fact, we have presented here a general version of the  phenomenon by using the setting of unambiguous quantum state discrimination under LOCC. In particular, we have shown that the degree of nonlocality  increases for both of the following situations: (i) reduction in the average entanglement content of the considered set, and  (ii) reduction in the number of entangled states within the set. Here, ``nonlocality'' is being used in the sense of a difference between global and local abilities to discriminate between shared quantum states of a set. And, by ``general version of the phenomenon'', we also mean that we have considered here the probabilistic version of the local quantum state discrimination problem within the phenomenon of more nonlocality with less entanglement, while in previous works, only the deterministic variety was considered.

The setting of unambiguous state discrimination under LOCC is found to be very useful in examining the degree of nonlocality of several sets which cannot be perfectly distinguished by LOCC. For the two-qubit case, we have discussed all possible combinations of states which may occur from the perspective of unambiguous local state discrimination.

\begin{acknowledgments}
We acknowledge support from the Department of Science and Technology, Government of India through the QuEST  grant (Grant No. DST/ICPS/QUST/Theme-3/2019/120).
\end{acknowledgments}

\bibliography{ref}
\end{document}